\documentclass[12]{article}
\usepackage{graphicx,amssymb,amsmath,geometry,amsthm}
\newtheorem{theorem}{Theorem}
\newtheorem{lemma}{Lemma}
\newtheorem{obs}{Observation}

\title{Maximum Area Rectangle Separating Red and Blue Points\thanks{This research was partially supported by NSF award IIP1439718 and CPRIT award RP150164}}

\author{
Bogdan Armaselu\thanks{Department of Computer Science, University of Texas at Dallas, {\tt bogdan.armaselu@utdallas.edu}}
\and
Ovidiu Daescu\thanks{Department of Computer Science, University of Texas at Dallas, {\tt  daescu@utdallas.edu}}
}
\date{}

\index{Armaselu, Bogdan}
\index{Daescu, Ovidiu}

\begin{document}
\thispagestyle{empty}
\maketitle
%\vspace{-0.2in}

\begin{abstract}
Given a set $R$ of $n$ red points and a set $B$ of $m$ blue points, we study the problem of finding a rectangle that contains all the red points, the minimum number of blue points and has the largest area. 
We call such rectangle a \textit{maximum separating rectangle}. 
We address the planar, axis-aligned (2D) version, and present an $O(m \log m + n)$ time, $O(m+n)$ space algorithm. The running time reduces to $O(m + n)$ if the points are pre-sorted by one of the coordinates. 
We further prove that our algorithm is optimal in the decision model of computation. 
\end{abstract}

\section{Introduction}

Consider two sets of points, $R$ and $B$ in the plane. $R$ contains $n$ points, called \textit{red points}, and $B$ contains $m$ points, called \textit{blue points}. 
The problem we study in this paper, called the \textit{Maximum Area Rectangle Separating Red and Blue Points} problem, is to find an axis-aligned rectangle that contains all the red points, the minimum number of blue points, and has the largest area. 
We call such a rectangle a \textit{maximum separating rectangle}. 
An example of a maximum separating rectangle is illustrated in Figures \ref{fig:example_ap}, for the planar axis-aligned case, as well as the planar non-axis-aligned case.

\begin{figure}[t]
\centering
\includegraphics[scale=0.5]{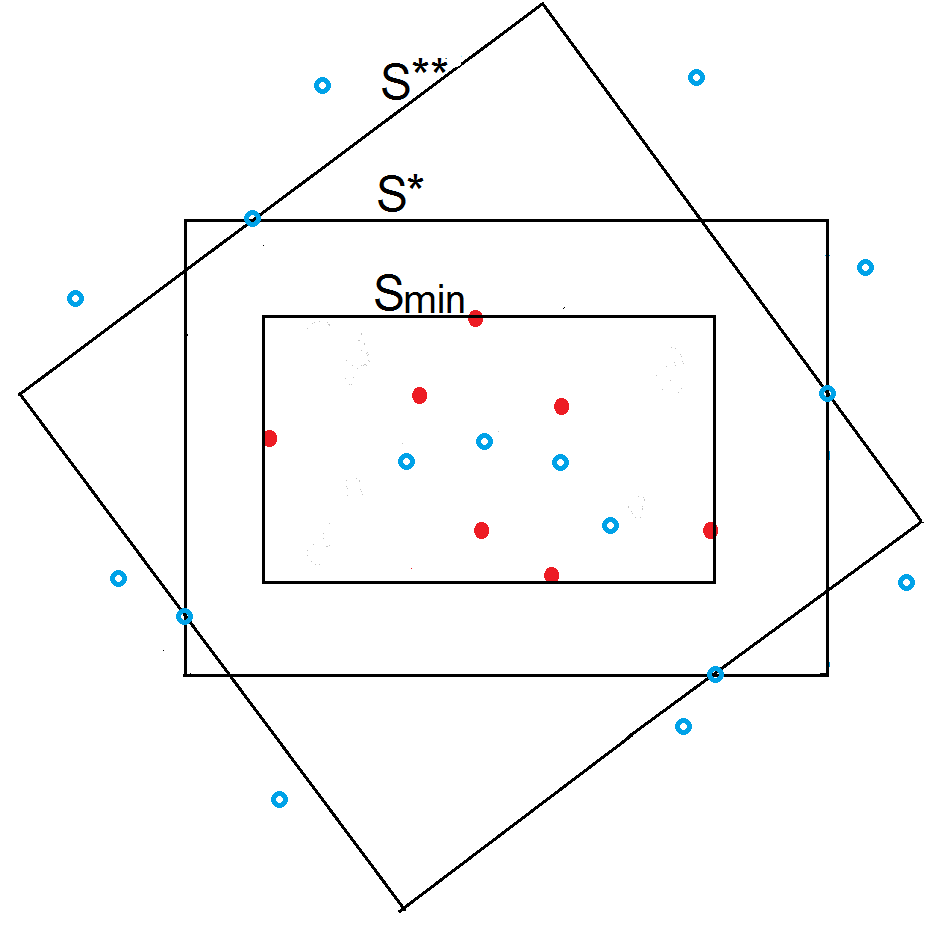}
\caption{Red points in $R$ are shown in solid circles and blue points in $B$ are shown in empty circles. 
The minimum enclosing rectangle $S_{min}$ of all red points, the maximum axis-aligned separating rectangle $S^*$, and the maximum non axis-aligned separating rectangle $S^{**}$, are shown.  }
\label{fig:example_ap}
\end{figure}

Applications that require separation of bi-color points could benefit from efficient results to this problem. 
For instance, consider we are given a tissue containing a tumor, where the tumor cells are specified by their coordinates. 
The coordinates of healthy cells are also given. The goal is to separate the tumor cells from the healthy cells, for surgical removal or radiation treatment. 
Another application would be in city planning, where blue points represent buildings and red points represent monuments and the goal is to build a park that contains the monuments and as few buildings as possible. 

\subsection{Related Work}
\label{Related Work} 

The problem of finding the largest-area empty axis-aligned rectangle among a set of $n$ points was introduced by Hsu et. al~\cite{Hsu}. 
They consider the axis-aligned version and show how to find all optimal solutions in $O(n^2)$ worst-case and $O(n \log^2 n)$ expected time. 
Later, Chazelle et. al showed how to find one optimal solution in $O(n \log^3 n)$ in the worst case \cite{Chazelle}. 
Currently, the best known result for this problem is by Aggarwal and Suri~\cite{Aggarwal}, namely $O(n \log^2 n)$ time nd $O(n)$ space to find an optimal solution.
To do that, they give a divide-an-conquer approach in which a largest empty corner rectangle problem is solved in the merging step.
They also show how to find the largest-perimeter empty rectangle optimally in $O(n \log n)$ time, by modifying their approach for finding the largest-area empty rectangle.

Mukhopadhyay et. al \cite{Mukhopadhyay} studied the problem of finding the maximum empty arbitrary oriented rectangle among a planar point set bounded by a given axis-aligned rectangle. 
They give an $O(n^3)$-time, $O(n^2)$-space algorithm to find all such maximum empty rectangles. 
Chaudhuri et. al \cite{Chaudhuri} independently gave a different solution for the same problem, also with $O(n^3)$ time and $O(n^2)$ space. 
For the 3D axis-aligned case, Nandy et. al \cite{Nandy} give an algorithm to compute the maximum empty box that runs in $O(n^3)$ time using $O(n)$ space.
%For the 3D axis aligned case, Datta et. al \cite{Datta} give an algorithm to compute the maximum empty box among $n$ points in the space that runs in $O(n^3)$ time using $O(n^2 \log n)$ space.

More recetntly, Dumitrescu and Jiang \cite{Dumitrescu} considered the problem of computing the maximum-volume axis-aligned $d$-dimensional box that is empty with respect to a given point set $P$ in $d$ dimensions. 
The target box is restricted to be contained within a given axis-aligned box $R$. 
For this problem, they give the first known FPTAS, which computes a box of a volume at least $(1 - \epsilon) OPT$, for an arbitrary $\epsilon$, where OPT is the volume of the optimal box. 
Their algorithm runs in $O((22 d \epsilon^{-2})^d \cdot n \log^d n)$ time.

Kaplan and Sharir study the problem of finding a maximal empty axis-aligned rectangle containing a query point \cite{Sharir}. 
They design an algorithm to answer queries in $O(\log^4 n)$ time, with $O(n \alpha(n) \log^4 n)$ time and $O(n \alpha(n) \log^3 n)$ space for preprocessing. 
Here $\alpha(n)$ is the inverse of the Ackermann's function. 
%It is worth noting that, when an "origin" point $O$ and staircases of points around $O$ are given, %they can find the largest rectangle containing only $O$ in $O(n \alpha(n))$ time \cite{Sharir, %Mozes}.
In a different paper, they also solve the disk version of the problem, tat asks to find the maximal empty disk containing a query point) \cite{Kaplan}. 
Their approach takes $O(\log^2 n)$ query and $O(n \log^2 n)$ preprocessing time, using a $O(n \log n)$ space data structure.

For the case where the input consists of axis-aligned rectangles rather than points, Nandy et. al \cite{Bhattacharya} presented a solution that finds the maximum empty rectangle in $O(n^2)$ worst-case, $O(n \log n)$ expected time.
In a follow-up paper \cite{Sinha}, they consider the problem of finding the largest empty rectangle among line segments and present an $O(n \log^2 n)$-time algorithm,
which they also extend to the case when the input consists of arbitrary polygonal obstacles.

Separability of two point sets using various separators is a well known problem. 
Megiddo et. al \cite{Megiddo} study the hyperplane separability of point sets in ${\mathbb R}^d$ in $d \geq 2$ dimensions. 
They show how to decide one-hyperplane separability using linear programming in polynomial time. 
They also prove that the problem of separating two point sets by $k$ lines is NP-complete. 
Aronov et. al \cite{Aronov} consider four metrics to evaluate misclassification of points of two sets by a linear separator. 
One of them is the number of misclassified points (which is somewhat related to our problem). 
For this error metric, they present an $O(n^d)$-time algorithm. 

Separating point sets with circles has also been considered. 
The problem of finding the minimum area separating circle among red and blue points was studied by Kosaraju et. al \cite{Kosaraju}. 
They solve the decision problem in linear time and give an $O(n \log n)$ time algorithm for computing the minimum separating circle, if it exists. 
When point sets cannot be separated by a circle, Bitner, Cheung and Daescu~\cite{Bitner} give two algorithms that run in $O(m n \log m + n \log n)$ and $O(m n + m^{1.5} \log^{O(1)} m)$ time, respectively.
Armaselu and Daescu~\cite{Armaselu} studied the dynamic version of this problem and presented results for the case when blue points are inserted or deleted. 

\subsection{Our Results}
\label{Our Results}

We consider the planar case of the problem and present an $O(m \log m + n)$ time, $O(m + n)$ space algorithm to find one optimal axis-aligned rectangle, based on a staircase approach. 
The running time reduces to $O(m + n)$ if the points are pre-sorted by one of the coordinates. 
We also prove a matching lower bound for the problem, by reducing from a known "$\Omega(n \log n)$-hard" problem. 

\section{Algorithms for the 2D axis-aligned version}
\label{2D-AA}

We begin by describing some properties of the bounded optimal solution.

\begin{obs}
\label{obs:2.1}
The maximum axis-aligned separating rectangle $S$ must contain at least one blue point on each of its sides. 
\end{obs}

Consider a quad $Q$ of 4 non-colinear blue points $q_1, q_2, q_3, q_4$ in this order. Two vertical lines going through two of these points and two horizontal lines going through the other two of those points define a rectangle $S$. 
We say that $Q$ \textit{defines} $S$. 

\bigskip

\noindent\textbf{Definition 2.1}. Consider a vertical strip formed by the two vertical lines bounding $R$ to the left and right, as well as a horizontal strip formed by the parallel lines bounding $R$ above and below. 
The \textit{minimum $R$-enclosing rectangle} $S_{min}$ is the intersection between the vertical and horizontal strips (refer to Figure \ref{fig:strip_ap}).

\bigskip

\noindent\textbf{Definition 2.2}. A \textit{candidate rectangle} is an $R$-enclosing rectangle that contains the minimum number of blue points and cannot be extended in any direction without introducing a blue point.

\bigskip

We start with the minimum $R$-enclosing rectangle $S_{min}$. For each side of $S_{min}$, we slide it outwards parallel to itself until it hits a blue point (if no such point exists, then the solution is unbounded). 
Denote by $S_{max}$ the resulting rectangle (shown in Figure \ref{fig:example_ap}). Unbounded solutions can be easily determined in linear time, so from now on we assume bounded solutions.
If the interior of $S_{max} \setminus S_{min}$ does not contain blue points, then $S_{max}$ is the optimal solution, and we are done. 
We discard the blue points contained in $S_{min}$, as well as the blue points outside of $S_{max}$, from $B$ and call the resulting set $B$.

$B$ is partitioned into 4 disjoint subsets (quadrants) $B_{NE}, B_{NW}, B_{SW}, B_{SE}$.
Each quadrant contains points that are located in a rectangle formed by right upper (resp. left upper, left lower, right lower) corners of $S_{min}$ and $S_{max}$ (see Figure \ref{fig:strip_ap} for details).

\begin{figure}[htp]
\centering
\includegraphics[scale=0.75]{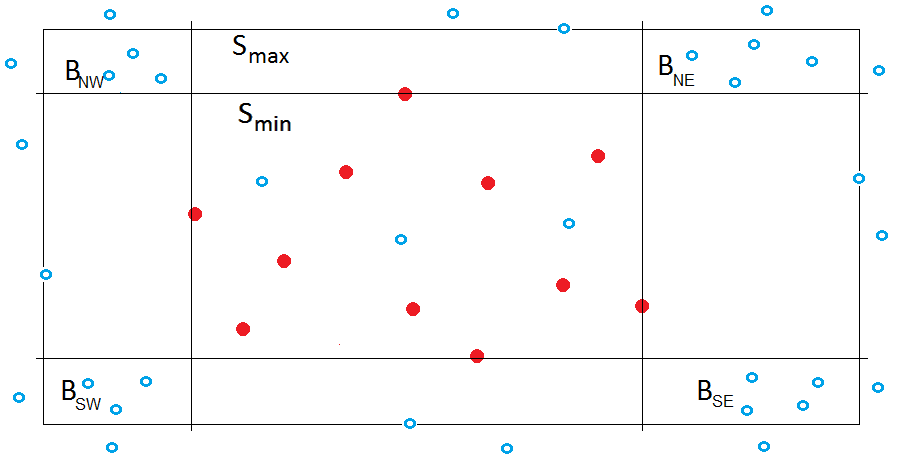}
\caption{The minimum enclosing rectangle $S_{min}$, the rectangle $S_{max}$ which bounds the solution space, and the subsets $B_{NE}, B_{NW}, B_{SW}, B_{SE}$. }
\label{fig:strip_ap}
\end{figure}

Consider the points of $B_{NE}$ and sort them by X coordinate. A point $p_0 = (x_0, y_0) \in B_{NE}$ is a candidate to be a part of an optimal solution only if there are no points $p = (x, y) \in B_{NE}$ with $x < x_0$ and $y < y_0$. 
We only leave such possible candidate blue points in $B_{NE}$ and discard the rest. 
The elements of $B_{NE}$ form a \textit{staircase} sequence, denoted $ST_{NE}$, which is ordered non-decreasingly by X coordinate and non-increasingly by Y coordinate, as shown in Figure \ref{fig:staircase}. 
The sets $B_{NW}, B_{SW}, B_{SE}$ are treated appropriately in a similar way and form the staircases $ST_{NW}, ST_{SW}, ST_{SE}$, which are ordered non-decreasingly by X coordinate. 
The 4 staircases can be found in $O(m \log m)$ time~\cite{Mukhopadhyay} and do not change in the axis-aligned cases.
Thus, abusing notation, we will refer to $ST_q$ as simply $B_q$ for each quadrant $q$.

While the staircase construction approach has been used before~\cite{Mukhopadhyay}, there are differences between how we use it in this paper and how it was used previously. 
Specifically, note that a maximum $B$-empty rectangle for a given orientation, computed as in \cite{Mukhopadhyay}, may not contain all red points.

\begin{figure}[htp]
\centering
\includegraphics[scale=0.75]{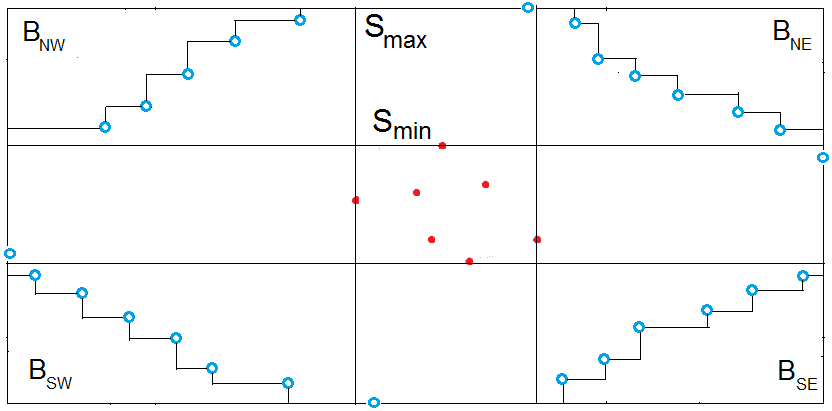}
\caption{The set of possible candidate points defining a solution. The sets $B_{NE}, B_{NW}, B_{SW}, B_{SE}$ are ordered by X, then by Y, and form a staircase}
\label{fig:staircase}
\end{figure}

\subsection{Finding all optimal solutions}
\label{Finding all optimal solutions}

We first prove that we can have $\Omega(m)$ maximum separating rectangles in the worst case. 
To do that, we use a construction similar to the one in~\cite{Jiang}, except that we have to ensure that all rectangles will contain $S_{min}$.

\begin{figure}[htp]
\centering
\includegraphics[scale=0.75]{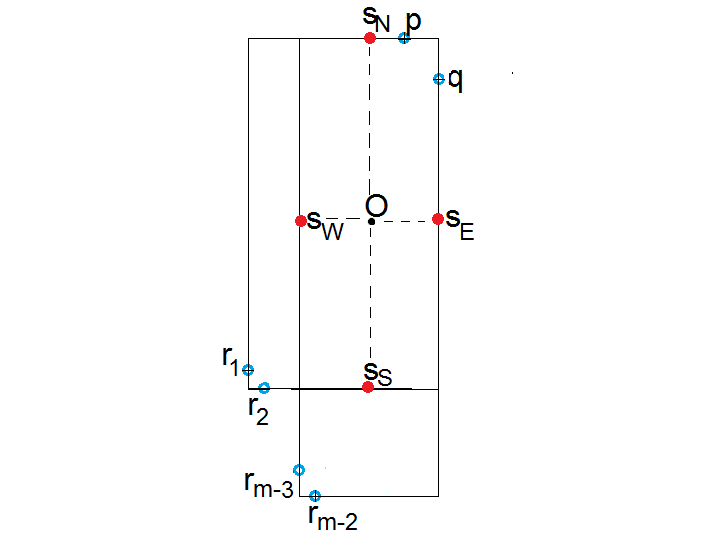}
\caption{There are 4 red points directly above, below, to the right, and to the left of the origin $O$.
The blue points are $p, q$, in $B_{NE}$, and $m-2$ other points in $B_{SW}$. 
All rectangles enclose $R$ and have the same area of $x_0 y_0$, thus giving $\Omega(m)$ maximum rectangles.}
\label{fig:lower_bound_aa}
\end{figure}

\begin{theorem}
\label{theorem:2.1}
In the worst case, there are $\Omega(m)$ maximum separating rectangles.
\end{theorem}

\begin{proof}
assume all blue points are in $B_{NE} \cup B_{SW}$ and refer to Figure \ref{fig:lower_bound_aa}. 
Let $B_{NE}$ be composed of two points $p, q$ and $B_{SW}$ of a sequence of points $r_1, \dots, r_{m-2}$, with the following coordinates.

1. $p(\frac{x_0}{4}, \frac{y_0}{2})$, with $x_0 > 0, y_0 > 0$;

2. $q(\frac{x_0}{2}, \frac{y_0}{4})$;

3. $r_i(x_i, y_i), \forall i = 1, \dots, m-2$;

4. $-\frac{3 x_0}{2} < x_i < x_j < 0, \forall i, j: 1 \leq i < j \leq m-2$;

5. $y_i = \frac{y_0}{2} - \frac{x_0 y_0}{\frac{x_0}{2} - x_{i-1}}, \forall i > 1$.

In addition, let $R$ contain 4 red points $s_E(\frac{x_0}{2}, 0), s_N(0, \frac{y_0}{2}), s_W(x_{m-3}, 0), s_S(0, y_2)$.

It is easy to check that all rectangles passing through $p, q, r_i, r_{i+1}$, for some $i, j: 1 < i < j \leq m-2$, enclose $R$.
Moreover, all these rectangles have an area equal to $(y_p - y_i) \cdot (x_q - x_{i-1}) = \frac{x_0 y_0}{\frac{x_0}{2} - x_{i-1}} \cdot (\frac{x_0}{2} - x_{i-1}) = x_0 y_0$.
All larger rectangles either contain a blue point or do not contain all red points.
Thus, there are $\Omega(m)$ maximum separating rectangles.
\end{proof}

To compute all maximum separating rectangles, we do the following.

We first compute $S_{min}$ and $S_{max}$ in $O(m + n)$ time. 
Then, we find all $B$-empty rectangles bounded by $S_{max}$ using the approach in~\cite{Hsu} in $O(m^2)$ time. 
For each such rectangle, we check whether it contains $S_{min}$ and, if it does not, we discard it.
Finally, we report the remaining rectangles that have the maximum area.

We have proved the following result.

\begin{theorem}
\label{theorem:2.2}
All axis-aligned maximum area rectangles separating $n$ red points and $m$ blue points can be found in $O(m^2 + n)$ time.
\end{theorem}

\subsection{Finding one optimal solution}
\label{Finding one optimal solution}

Suppose $S_{min}$ and the four staircases have been already computed. We describe an algorithm to find only one optimal solution. 

Observe that all maximal rectangles containing $S_{min}$ are defined by tuples of four points from $B_i, i = NE, NW, SW, SE$ (called \textit{support points}). 
Each of these points supports an edge of the rectangle.
For the $k$-th candidate rectangle, denote the top support by $top_k$, the left support by $left_k$, the right support by $right_k$ and the bottom support by $bottom_k$.

The problem that we solve is essentially the one of finding the maximal $B$-empty rectangle containing a given "origin" point, considered in~\cite{Sharir, Mozes}.
However, in our case, the target rectangle has to contain $S_{min}$, rather than a single point. 
Based on the position of each support of a candidate rectangle, 3 cases may arise~\cite{Sharir, Mozes}.
Note that the solution in~\cite{Sharir} takes $O(m \log m)$ time in cases 1 and 2 and $O(m \alpha(m))$ time in case 3. 
We show how to solve this problem in $O(m)$ time in every case.

\textbf{Case 1}. Three supports are in the same side of $S_{min}$ and the fourth support is on the opposite side of $S_{min}$.

Suppose without loss of generality (wlog) that the top, right and bottom supports lie to the right of $S_{min}$ and the left one lies to the left of $S_{min}$ (as in Figure \ref{fig:candidate_rectangles_1}).
Note that, for each top-right tuple with the top support in $B_{NE}$, there is a unique bottom support to the right of $S_{min}$, which is in $B_{SE}$. 
Thus, the left support is also unique. 
As argued in~\cite{Sharir,Mozes}, there are $O(m)$ top-right tuples.
Similarly, for each top-left tuples with the top support in $B_{NW}$, we get a unique bottom-right tuple.
This gives us a total of $O(m)$ candidate rectangles in case 1. 
See \cite{Sharir, Mozes} for more details.

\begin{figure}[htp]
\centering
\includegraphics[scale=0.75]{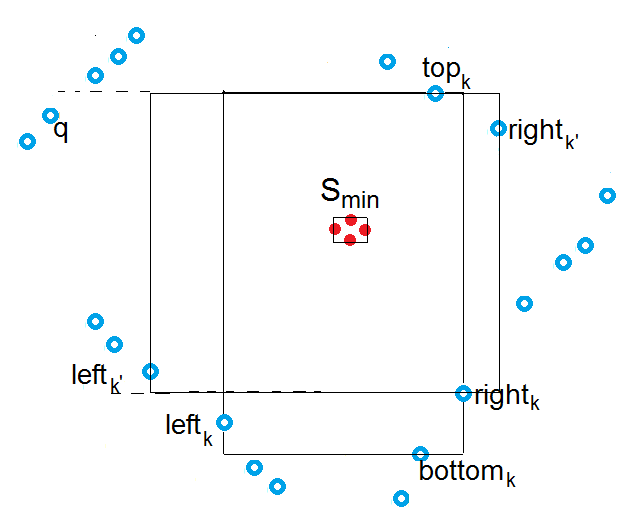}
\caption{The top, right, and bottom supports lie to the right of $S_{min}$ and the left support lies to the left of $S_{min}$.
	For any such top-right pair $(top_k, right_k)$, there is a unique bottom support $bottom_k$ to the right of $S_{min}$ and a unique left support $left_k$ to the left of $S_{min}$.}
\label{fig:candidate_rectangles_1}
\end{figure}

\textbf{Case 2}. Each support is from a different quadrant.

Suppose wlog that $top_k \in B_{NE}$, which implies $right_k \in B_{SE}, bottom_k \in B_{SW}$ and $left_k \in B_{NW}$ (see Figure \ref{fig:candidate_rectangles_2}).
For each top-right tuple satisfying these conditions, there is a unique bottom support from $B_{SW}$ and a unique left support from $B_{NW}$. 
Again, as argued in~\cite{Sharir,Mozes}, there are $O(m)$ top-right tuples.
Similarly, for each top-left tuple with $top_k \in B_{NW}$, the bottom and right supports satisfying the condition are unique.
Thus, there are $O(m)$ candidate rectangles in Case 2. 
More details can be found in \cite{Sharir, Mozes}.

\begin{figure}[htp]
\centering
\includegraphics[scale=0.75]{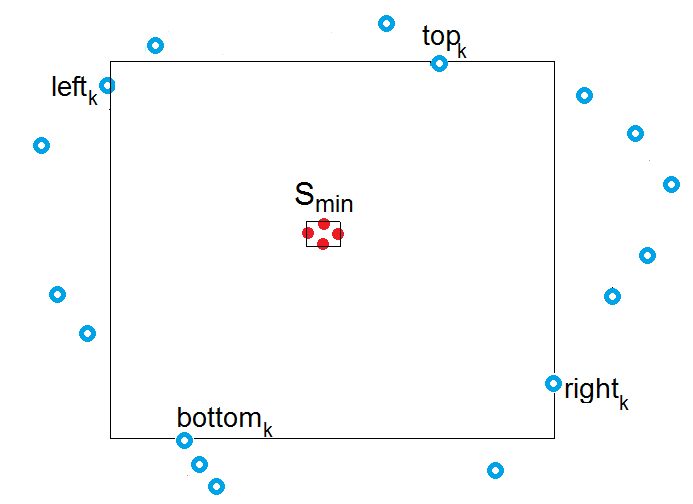}
\caption{Each support is from a different quadrant, with $top_k \in B_{NE}$.
	For any such top-right pair $(top_k, right_k)$, there is a unique bottom support $bottom_k \in B_{SW}$ and a unique left support $left_k \in B_{NW}$.}
\label{fig:candidate_rectangles_2}
\end{figure}

\textbf{Case 3}. Two supports are from a quadrant and the other two are from an opposite quadrant.

Suppose wlog that $top_k, right_k \in B_{NE}$ and $bottom_k, left_k \in B_{SW}$ (refer to Figure \ref{fig:candidate_rectangles_3}).
For each such top-right pair, there are multiple choices of bottom-left pairs formed by adjacent points from $B_{SW}$. 
However, the bottom support has to be above or equal to the last point $p \in B_{SE}$ to the left of $right_k$, if $p$ exists, 
and the left support has to be to the right or equal to the last point $q \in B_{NW}$ below $top_k$, if $q$ exists.
Consider two functions, $f$ and $l$, that assign, to every top-right pair $(top_k, right_k)$, 
the index in $B_{SW}$ of the first (resp., last) bottom support that occurs with $(top_k, right_k)$, denoted by $f(k)$ and $l(k)$, respectively.
Note that $f$ and $l$ are monotonically decreasing functions (as shown in Figure \ref{fig:candidate_rectangles_3}). 

\begin{figure}[htp]
\centering
\includegraphics[scale=0.75]{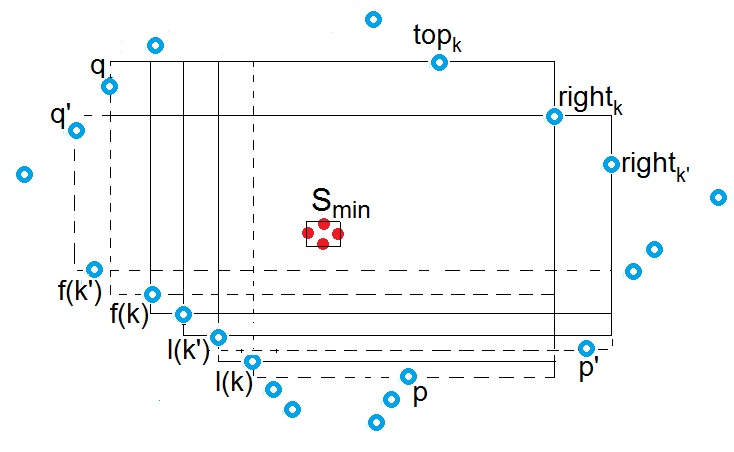}
\caption{For each top-right pair $(top_k, right_k) \in B_{NE}^2$, there are multiple bottom-left pairs $(bottom_k, left_k) \in B_{SW}^2$. 
However, they have to lie above $p$ (if $p$ exists), and to the right of $q$ (if $q$ exists).
For the next top-right pair $(top_{k'} = right_k, right_{k'}) \in B_{NE}^2$, the bottom-left pairs have to lie above $p'$ (if $p'$ exists) and to the right of $q'$ (if $q'$ exists). 
Note that $p'$ is after $p$ in $B_{SE}$ and $q'$ is before $q$ in $B_{NW}$.}
\label{fig:candidate_rectangles_3}
\end{figure}

Cases 1 and 2 give $O(m)$ tuples and we will argue later how to find all these tuples in $O(m)$ time.
Case 3 gives $O(m^2)$ tuples so, from now on, we focus on case 3.
Note that all supporting points are in $B_{NE} \cup B_{SW}$.
Candidate rectangles are defined by two pairs of adjacent points in $B_{NE}$ and $B_{SW}$ respectively.
Denote by $(p_i, q_i)$, the $i$-th pair of adjacent points in $B_{NE}$, and by $(r_j, s_j)$, the $j$-th pair of adjacent points in $B_{SW}$. 
Consider a matrix $M$ such that $M(i, j)$ denotes the area of the rectangle supported on the top-right by $(p_i, q_i)$ and on the bottom-left by $(r_j, s_j)$.
Some entries $(i, j)$ of $M$ correspond to cases where $(p_i, q_i)$ and $(r_j, s_j)$ do not define a $B$-empty rectangle 
and are therefore set to "undefined".
The goal is to compute the maximum of each row $i$, along with the column $j(i)$ where it occurs. To break ties, we always take the rightmost index.

It is easy to see that the defined portion of each row of $M$ is contiguous. 
Since the functions $f$ and $l$ are monotonically decreasing, the defined portion of $M$ has a staircase structure (as shown in Figure \ref{fig:staircase_matrix}).
We say that $M$ is \textit{staircase-defined} by $f$ and $l$.
Moreover, it turns out that $M$ is a partially defined inverse Monge matrix~\cite{Sharir} of size $O(m) \cdot O(m)$, and thus all row-maxima can be found in $O(m \alpha(m))$ time using the algorithm in~\cite{Klawe} (see also~\cite{Mozes}).

\begin{figure}[htp]
\centering
\includegraphics[scale=0.75]{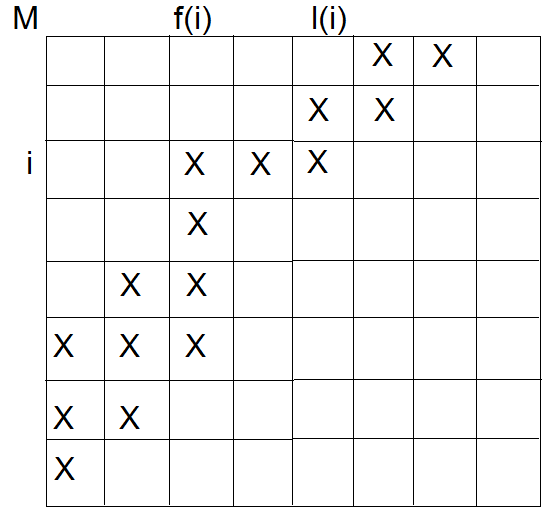}
\caption{Staircase matrix $M$. The defined portion is marked by X's.}
\label{fig:staircase_matrix}
\end{figure}

To eliminate the $\alpha(m)$ factor, we extend $M$ to a totally inverse monotone matrix, so that all row maxima can be found in $O(m)$ time using the SMAWK algorithm~\cite{Smawk}.
Recall that $M$ is totally inverse monotone iff $\forall i, j, k, l$ such that $1 \leq i < k \leq m, 1 \leq j < l \leq m$, we have $M(k, j) \leq M(k, l) => M(i, j) \leq M(i, l)$.
To do that, one could try the approach in~\cite{Sinha} to make $M$ totally inverse monotone, by filling it with 0's on both sides of the defined portion. 
However, it does not work in our case. 
Note that it may happen that there exist $i, j, k, l: 1 \leq i < k \leq m, 1 \leq j < l \leq m$ such that $M(i, j) > 0$ and $M(i, l) = M(k, j) = M(k, l) = 0$, so $M$ is not totally inverse monotone.
Moreover, it follows by a similar argument that $M$ is not totally monotone either.
Therefore, we resort to a different filling scheme, which is similar to the one in~\cite{Klawe}.
Specifically, we fill only the left undefined portion of each row with 0. 
That is, if the defined portion of row $i$ starts at $j_1(i)$, then $M(i, j) = 0, \forall 1 \leq i \leq m, 1 \leq j < j_1(i)$.
We also fill the right undefined portion of each row with negative numbers such that, 
if the defined portion of row $i$ ends at $j_2(i)$, then $M(i, j) = j_2(i) - j, \forall 1 \leq i \leq m, j_2(i) < j \leq m$, 
that is, negative numbers in decreasing order (see Figure \ref{fig:totally_monotone_matrix}).

\begin{figure}[htp]
\centering
\includegraphics[scale=0.75]{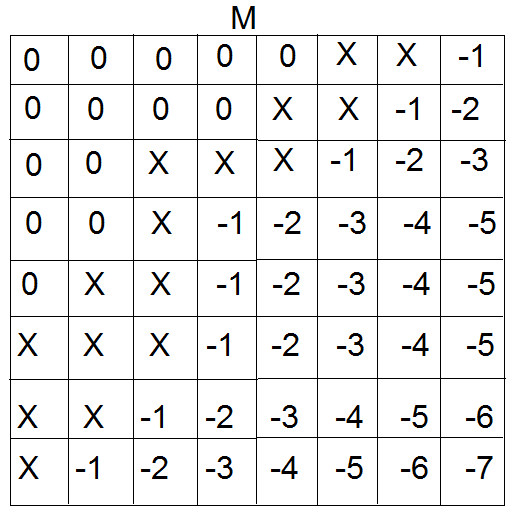}
\caption{$M$ is padded on the left of the defined portion with 0's, and on the right of the defined portion with negative numbers in decreasing order.}
\label{fig:totally_monotone_matrix}
\end{figure}

\begin{lemma}
\label{lemma:2.3}
$M$ is a totally inverse monotone matrix.
\end{lemma}

\begin{proof}
Suppose this is not the case. 
Then there exist $i, j, k, l, 1 \leq i < k \leq m, 1 \leq j < l \leq m$ such that $M(i, j) > M(i, l)$ and $M(k, j) \leq M(k, l)$.
If $M(i, l) < 0$, then by construction $M(k, l) < 0$, so $M(k, j) > M(k, l)$, contradiction. 
Thus, $M(i, l) \geq 0$. 
If $M(i, l) = 0$, then by construction $M(i, j) = 0$, again contradiction. 
Hence, $M(i, j) > M(i, l) > 0$, which entails $M(k, j) > M(k, l)$ or $M(k, j) > 0, M(k, l) > 0$.
The first choice gives a contradiction, so the only remaining possibility is $M(i, j) > M(i, l) > 0$ and $0 < M(k, j) \leq M(k, l)$. 
But this contradicts the total inverse monotonicity of the defined (positive) portion of $M$.
\end{proof}

As a side note, if the functions $f$ and $l$ were monotonically increasing, rather than decreasing, we could make $M$ totally monotone instead. 
That is, $M(i, j) < M(i, l) => M(k, j) < M(k, l)$.
To do that, we fill the "undefined" portion of $M$ with zeros on the right side, and negative numbers in increasing order on the left side of the defined portion of $M$.
By a similar argument as in Lemma \ref{lemma:2.3}, it follows that $M$ is totally monotone.

Note that computing $M$ explicitly would take $\Omega(m^2)$ time. 
To avoid that, we only store the pairs from $B_{SW}$ that may define optimal solutions in a list, as in~\cite{Smawk}, and evaluate $M(i, j)$ only when needed. 
Thus, we only require $O(m)$ time and space. 

\begin{figure}[htp]
\centering
\includegraphics[scale=0.75]{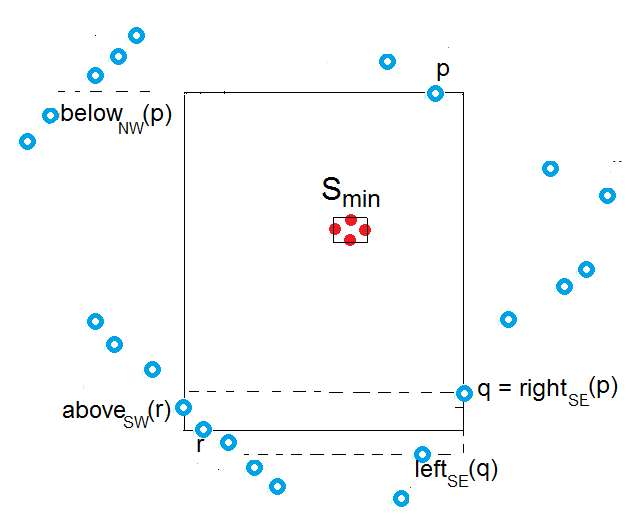}
\caption{The pointers $above$, $below$, $left$ and $right$.}
\label{fig:candidate_rectangle}
\end{figure}

In order to find the optimal solution, after having computed the staircases, we do the following.
Suppose we have fixed a top-right pair $(p_k, q_k)$, with $p_k \in B_{NE} (B_{NW}), q_k \in B_{NE} \cup B_{SE} (B_{NE} \cup B_{SW})$.
The leftmost possible support is the highest point in $B_{NW}$ below $p_k$, denoted $below_{NW}(p_k)$, and the lowest possible support is the rightmost point in $B_{SE}$ to the left of $q_k$, denoted $left_{SE}(q_k)$.
Both such extremal supports can be found in $O(1)$ time if we store, with each point $p$ in some staircase, and each quadrant $q$, the pointers $below_q(p)$ and $left_q(p)$  (see Figure \ref{fig:candidate_rectangle}). 
In addition, we consider, for each $r$, the pointers $above_q(r)$ to the lowest point in quadrant $q$ above $r$, and $right_q(p)$ to the leftmost point in quadrant $q$ to the right of $p$.
These pointers can be precomputed in $O(m)$ time for all $p \in B$ through a scan in X order.
Note that these pointers are defined in a similar manner as in the algorithm in \cite{Aggarwal} for finding the largest empty corner rectangle.
Also, the staircases are stored as doubly-linked lists with pointers $prev(p), next(p)$ to the point before (resp., after) $p$ in the respective staircase, plus the pointers mentioned above whenever needed.	
We then consider all bottom-left pairs occurring with $(p_k, q_k)$ and assume wlog that $p_k \in B_{NE}$.
If $p_k, q_k$ and one of $left_{SE}(q_k), bottom_{NW}(p_k)$ (say $left(q_k)$) are on the same side of $S_{min}$,  then we are in case 1 with a rectangle defined by $p_k, q_k, left(q_k), above(left(q_k))$.
If $p_k, q_k, left_{SW}(q_k)$, and $below_{NW}(p_k)$ are on different quadrants, then we are in case 2 with a rectangle defined by $p_k, q_k, left_{SW}(q_k), below_{NW}(p_k)$.
Thus, cases 1 and 2 require $O(m)$ time and space in total.
Otherwise (i.e. we are in case 3), we store $f(k) = right_{SW}(below_{NW}(p_k)), l(k) = above_{SW}(left_{SE}(q_k))$, respectively, in two arrays, $F$ and $L$.
After all top-right and top-left pairs are treated, we run the SMAWK algorithm as described earlier, in order to compute all row-maxima of $M$ in $O(m)$ time, 
where $M$ is staircase matrix defined by $F$ and $L$, and $M(i, j)$ is the area of rectangle defined by the $i$-th pair from one quadrant and the $j$-th pair from the opposite quadrant in case 3.
We then report the (last index) rectangle corresponding to the maximum between all row-maxima of $M$ and all maximum area rectangles obtained in cases 1 and 2, along with its area.
Thus, we have proved the following result.

\begin{theorem}
\label{theorem:2.4}
The axis-aligned version of the maximum-area separating rectangle problem can be solved in $O(m \log m + n)$ time and $O(m + n)$ space.
The running time reduces to $O(m + n)$ if the blue points are presorted by their X coordinates.
\end{theorem}

\subsection{Lower bound}
\label{Lower bound}
In this section, we prove that $\Omega(m \log m + n)$ steps are sometimes needed 
in order to compute a maximum axis-aligned separating rectangle, provided that the blue points are not pre-sorted.

To do that, we reduce our problem from the 1D-Furthest-Adjacent-Pair problem, 
which is known to have a lower bound of $\Omega(m \log m)$ for a set of $m$ numbers.

In the 1D-Furthest-Adjacent-Pair, we are given a set $A$ of $m$ real numbers, 
and the goal is to find the two numbers $a, b \in A, a < b$ for which the quantity $b - a$ is the maximum among all adjacent pairs in the sorted order of $A$ (denoted by $A'$).
Wlog assume that all these numbers are in the interval $[0, 1]$.

The reduction to our problem is as follows.
For each $a_i \in A$, we consider two points, $p_i(a_i, \frac{1}{1 + a_i})$ and $q_i(-a_i, -\frac{1}{1 + a_i})$.
The set of blue points $B$ is the set of all such $p_i$'s and $q_i$'s.
The red point set $R$ consists of the origin $O$, together with four special points, $s_E(a_2, 0), s_N(0, \frac{1}{1 + a_{m-1}}), s_W(-a_2, 0)$, and $ s_S(0, -\frac{1}{1 + a_{m-1}})$.
See Figure \ref{fig:lower_bound_aa_1} for details.

\begin{figure}[htp]
\centering
\includegraphics[scale=0.75]{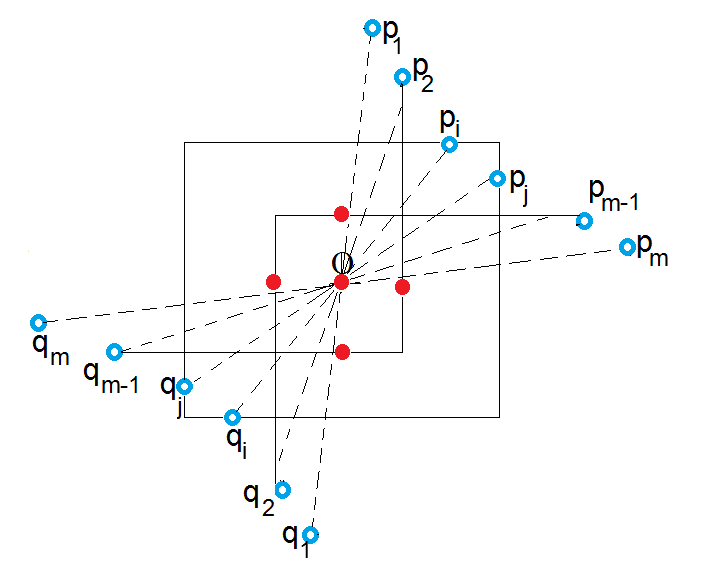}
\caption{$R$ consists of the origin $O(0, 0)$ and four points $s_E, s_N, s_W, s_S$.
		For each $a_i \in A$, there are two blue points $p_i, q_i$ of opposite coordinates. 
		Each tuple $(p_i, p_j, q_i, q_j), 1 \leq i, j \leq m$ defines a candidate rectangle.}
\label{fig:lower_bound_aa_1}
\end{figure}

Now we prove that the reduction works.

\begin{lemma}
\label{lemma:2.5}
Two adjacent numbers $a_i, a_j \in A$ form the furthest adjacent pair of numbers if and only if 
the rectangle $S$ bounded by $p_i, p_j, q_i$ and $q_j$ is one of the separating rectangles of maximum area.
\end{lemma}

\begin{proof}
Let $a_i, a_j$ be adjacent in $A'$ and $a_j - a_i \geq |a_k - a_l|, \forall a_k, a_l \in A$ that are adjacent in $A'$.
Note that $y(p_i)$ decreases while $x(p_i)$ increases (same for $q_i, p_j, q_j$) and $y(p_i) \geq s_N, x(p_j) \geq s_E, y(q_i) \leq s_S, x(q_j) \leq s_W, \forall 1 \leq i < j \leq m$, so $S$ encloses $R$.
We need to show that the rectangle $S$ has the following properties:

(1) does not contain any blue point, and

(2) has the maximum area among all such rectangles.

If $S$ would contain a point $p_k$ (similarly, $q_k$), then we would have $x(p_k) = a_k$ with $a_i < a_k < a_j$, contradiction. 
So property (1) is satisfied.
We have the points $p_i(a_i, \frac{1}{1 + a_i}), q_i(-a_i, -\frac{1}{1 + a_i}), p_j(a_j, \frac{1}{1 + a_j})$ and $q_j(-a_j, -\frac{1}{1 + a_j})$.
This means that $area(S) = 4 \frac{a_j}{1 + a_i}$. 
Let $a_j - a_i = d > 0$ and $f(d) = area(S)$. 
We have $f(d) = 4 \frac{a_i + d}{a_i + 1}$, which is monotonically increasing.
Therefore, $S$ has the highest area among all rectangles with property (1).

Now let $S$ be a largest separating rectangle bounded by $(p_i, p_j, q_j, q_i)$, for some $a_i, a_j$.
Since $S$ does not contain any blue points, there cannot exist any $a_k \in A: a_i < a_k < a_j$, so $a_i, a_j$ are adjacent in $A'$.
Moreover, since $area(S)$ is monotonically increasing in $a_j - a_i$, it follows that $a_j - a_i \geq |a_k - a_l|, \forall a_k, a_l \in A$ that are adjacent in $A'$.
\end{proof}

This gives us the following result.

\begin{theorem}
\label{theorem:2.6}
$\Omega(m \log m + n)$ steps are needed in order to compute the maximum axis-aligned rectangle separating $R$ and $B$.
\end{theorem}

\begin{proof}
First, $\Omega(n)$ steps are needed in order to compute $S_{min}$, 
as one cannot compute the optimal solution without knowing $S_{min}$.
The $\Omega(m \log m)$ term follows from the reduction from 1D-Furthest-Pair. 
\end{proof}

\section{Conclusion and Remarks}
\label{Conclusion and Remarks}

We addressed the problem of finding the maximum area axis-aligned separating rectangle that encloses all red points and the minimum number of blue points proved a lower bound for the problem, and provided optimal algorithms. 

We reopen the problem of finding a maximum area empty rectangle among points in the plane, Specifically, either prove an $O(n \log^2n)$ lower bound (which seems unlikely), or improve over the thirty years old $O(n \log^2n)$ time algorithm of Aggarwal and Suri~\cite{Aggarwal}.  

We also leave open whether it is possible to adapt the maximum empty rectangle containing a query point data structures in~\cite{Sharir, Mozes}
to find the maximum area rectangle containing a set of red points and fewest number of blue points when the red points are given at query time.
Notice that this version could have important applications, including in fabrication of integrated circuits, where red points could represent defects in the fabrication boards. 

\section*{Acknowledgement}
\label{Acknowledgement}
The authors would like to thank Dr. Anastasia Kurdia for her preliminary work on the maximum separating rectangle problem.

%---------------------------- Bibliography -------------------------------

% Please add the contents of the .bbl file that you generate,  or add bibitem entries manually if you like.
% The entries should be in alphabetical order
\small
\bibliographystyle{abbrv}

\end{document}